\documentclass[ASNA,twocolumn]{USG} 
\usepackage{anyfontsize} %

\usepackage{colortbl}   
\usepackage{xcolor}     

\usepackage{steinmetz}
\graphicspath{{./images/}}

\articletype{}%
\subarticletype{}

\received{}
\revised{}
\accepted{}
\volume{}
\copyyear{}
\startpage{}

\begin{document}
\title{Deep Neural Mixed Models for Nonlinear Longitudinal Prediction with Subject-Specific Heterogeneity}
\transtitle{Guidelines for Establishing a Cytometry Laboratory}
\subtranstitle{trans-subtitle}
\author[1]{Mahmud Hasan}[https://orcid.org/0000-0003-3338-8163]
\author[2]{Niloofar Ramezani}[ https://orcid.org/0000-0001-8941-4696]
\author[3]{Hrishikesh Chakraborty}[https://orcid.org/0000-0001-9078-845X]





















\address[1]{\orgdiv{Department of Mathematics and Statistics, }\orgname{Penn State Altoona,  }%
\orgaddress{\state{PA, }\country{USA}}}

\address[2]{\orgdiv{Department of Biostatistics, }\orgname{Virginia Commonwealth University, }%
\orgaddress{\state{VA, }\country{USA}}}

\address[3]{\orgdiv{Department of Biostatistics \& Bioinformatics, }\orgname{Duke University, }%
\orgaddress{\state{NC, }\country{USA}}}

\corres{Mahmud Hasan  (\email{mnh5524@psu.edu})}


\transkeywords{Longitudinal data| Deep learning| Mixed-effects models| Neural networks| Random effects| Prediction| DGNMM.}
\abstract[ABSTRACT]{
Longitudinal health data often exhibit within-subject correlation,
subject-specific heterogeneity, and complex nonlinear predictor--response
relationships that may not be adequately captured by conventional mixed
models. We propose a Deep Generalized Neural Network Mixed Model (DGNMM)
that extends the GLMM framework by using a deep neural network to learn
nonlinear population-level relationships while incorporating directly
trainable, regularized subject-specific effects to capture individual
heterogeneity. The model is estimated using prediction-oriented
stochastic optimization with separate regularization of the
neural-network parameters and subject-specific effects. This formulation
allows us to investigate when increased nonlinear representation capacity
provides additional predictive benefit in the presence of subject-specific
heterogeneity. A simulation study varies population-level nonlinearity
and subject-specific heterogeneity and compares DGNMM with GLMM,
artificial neural network (ANN), Maity--Pal (MP), and GNMM alternatives
for continuous and binary longitudinal outcomes. In the real-data
application, DGNMM achieved the lowest MSPE for continuous depression
and anxiety prediction, while alternative models remained competitive
for binary outcomes. Exploratory subject-level analyses showed
heterogeneous predictive gains across individuals, but the magnitude of
the estimated subject-specific effects was not significantly associated
with predictive gain. Overall, increased nonlinear representation
capacity was most beneficial for continuous prediction under nonlinear
response structures, while greater subject-specific heterogeneity did
not systematically yield larger predictive gains.
}


\maketitle


\section{Introduction}\label{sec1}

Longitudinal and clustered data arise frequently in biostatistics, epidemiology, precision medicine, public health, and healthcare analytics. In such studies, repeated measurements are collected from the same subject over time. As a result, observations within a subject are correlated and require appropriate statistical methods to account for this correlation and ensure valid inference and accurate prediction. Classical approaches for analyzing longitudinal data, including linear mixed models (LMMs) \citep{Verbeke1996} and generalized linear mixed models (GLMMs) \citep{breslow1993approximate}, provide a flexible framework for modeling correlated outcomes through subject-specific random effects. These models have been extensively applied to longitudinal studies involving continuous, binary, and count outcomes. In many modern applications, digital health technologies, wearable sensors, electronic health records, and mobile health (mHealth) platforms have generated large-scale longitudinal datasets containing behavioral, physiological, and clinical information \citep{steinhubl2015digital,torous2016new}. These data provide new opportunities for predicting future health outcomes, including disease status and symptom severity such as depression and anxiety. However, the relationships between predictors and outcomes are often highly nonlinear and may not be adequately captured by linear fixed-effect models such as GLMMs \citep{breslow1993approximate}. One important application arises in schizophrenia research, where patients are monitored over time using smartphone-based digital phenotyping technologies \citep{torous2016new}. Data collected from GPS, accelerometers, call logs, text messaging activity, and other mobile sensors provide valuable information about behavioral patterns that may be associated with depression, anxiety, and other mental health outcomes \citep{onnela2021opportunities}. Accurate prediction of such outcomes requires modeling approaches that can simultaneously account for within-subject dependence, subject-specific heterogeneity, and complex nonlinear relationships between predictors and responses. More accurate longitudinal prediction may facilitate earlier identification of symptom worsening and support timely clinical decision-making for schizophrenia patients.

Deep learning methods, such as artificial neural networks (ANNs), are powerful predictive models for capturing complex nonlinear relationships between predictors and responses \citep{goodfellow2016deep}. However, standard ANN formulations generally do not explicitly account for dependence among repeated observations from the same subject. Several approaches have been proposed for modeling nonlinear relationships and correlated longitudinal data. Recurrent neural networks (RNNs) have been widely used for sequential prediction tasks \citep{choi2016using,choi2016doctor}, while mixed-effects neural network models have been developed to incorporate subject-specific dependence \citep{wortwein2023neural,kock2025deep}. Neural network approaches for repeated measures data based on subject-specific indicator variables have also been proposed \citep{maity2013repeated}, and standard ANNs have been applied to longitudinal prediction problems without explicitly modeling within-subject dependence \citep{ganesan2014prediction}. Beyond neural networks, nonlinear mixed-effects models \citep{lindstrom1990nonlinear,davidian1995nonlinear}, generalized additive models (GAMs) \citep{hastie1990generalized}, multivariate adaptive regression splines (MARS) \citep{friedman1991multivariate}, and classification and regression trees (CART) \citep{breiman1984classification} provide flexible alternatives for modeling nonlinear predictor--response relationships. Despite these advances, an important prediction question remains regarding how increased nonlinear representation capacity can be combined with subject-specific heterogeneity in longitudinal data, and when such additional flexibility provides meaningful predictive benefit.
Other neural mixed-model approaches incorporate random effects within neural network frameworks \citep{tandon2006neural,tran2017random}.

In particular, \citet{mandel2023neural} proposed the Generalized Neural Network Mixed Model (GNMM), which extends the GLMM framework \citep{breslow1993approximate} by replacing the linear population-level component with a feed-forward neural network while retaining
subject-specific random effects. The GNMM framework can accommodate
nonlinear predictor--response relationships while accounting for
within-subject dependence, and its general formulation can accommodate
feed-forward neural networks of arbitrary size, although its empirical
investigation focused primarily on shallow one- and two-hidden-layer
networks. Thus, the motivation for further increasing neural-network
complexity is not simply that nonlinear relationships cannot be
represented by GNMM. Rather, an important prediction question is whether
increased nonlinear representation capacity provides meaningful
improvement as the underlying response surface becomes progressively more
complex, and how such improvement changes with the magnitude of
subject-specific heterogeneity.

A further distinction concerns estimation. GNMM is formulated within a
mixed-model framework using likelihood- or quasi-likelihood-based
estimation of the neural-network and random-effect components. Such estimation may become computationally demanding as the complexity of the population-level neural network increases. This motivates consideration of a prediction-oriented formulation in
which the neural-network parameters and regularized subject-specific
latent effects are estimated jointly using stochastic optimization. Such
a formulation provides a computationally direct approach for combining
increased nonlinear representation capacity with subject-specific
heterogeneity in longitudinal prediction.

In this paper, we propose a Deep Generalized Neural Network Mixed Model
(DGNMM) for nonlinear longitudinal prediction in the presence of
subject-specific heterogeneity. The DGNMM decomposes the conditional
response into two components: a population-level nonlinear response
surface, represented by a multilayer feed-forward neural network, and a
trainable subject-specific latent effect shared across repeated
observations from the same individual. The neural-network weights and
subject-specific effects are estimated jointly using stochastic
gradient-based optimization. Separate $L_2$ penalties are imposed on the
network weights and subject-specific effects, allowing the complexity of
the population-level response surface and the degree of subject-specific
shrinkage to be controlled separately. The framework accommodates both continuous and binary longitudinal outcomes and is evaluated for within-subject prediction, where previous observations from an individual are used to estimate the subject-specific effect before predicting a held-out future observation.

The main contributions of this paper are as follows. First, we develop a
prediction-oriented DGNMM that combines increased nonlinear representation
capacity with regularized subject-specific latent effects and enables their
joint estimation through stochastic optimization. Second, we introduce
separate regularization of the population-level neural-network parameters
and subject-specific effects. We further provide a statistical
interpretation of the latter by showing that, in a Gaussian linear special
case, the quadratic penalty on the subject-specific effect has the same
shrinkage form as that induced by a Gaussian random intercept.

Third, we introduce a controlled simulation framework that independently
varies population-level structure from linear to increasingly nonlinear
settings and the magnitude of subject-specific heterogeneity. This design
allows us to investigate when increased nonlinear representation capacity
improves prediction, when simpler models remain preferable, and how these
patterns change with subject-specific heterogeneity. We consider both continuous and binary outcomes and compare DGNMM with GLMM, ANN, MP, and shallower GNMM alternatives using MSPE and AUC, respectively. Finally, we evaluate DGNMM using longitudinal smartphone digital-phenotyping
data from individuals with schizophrenia to predict depression and anxiety.
In addition to comparing overall predictive performance, we examine the
estimated subject-specific effects and assess whether their magnitudes are
associated with subject-level predictive gains relative to the main-effects
GLMM. Together, the simulation and real-data analyses provide complementary
evidence on the roles of nonlinear representation capacity and
subject-specific heterogeneity in longitudinal prediction. In the schizophrenia application, improved symptom prediction may support
earlier detection of symptom worsening and better clinical monitoring\citep{kahn2015schizophrenia,millan2016altering}.

The remainder of the paper is organized as follows.
Section~\ref{sec2} reviews existing models and introduces the proposed
DGNMM framework. Section~\ref{sec3} presents the estimation procedure,
prediction methodology, regularization interpretation, and performance evaluation criteria. Section~\ref{sec4} presents the controlled simulation study and comparative prediction results across different levels of nonlinear complexity and subject-specific heterogeneity.
Section~\ref{sec5} presents the schizophrenia symptom prediction
application, including the subject-specific heterogeneity analysis. Finally, Section~\ref{sec6} concludes with a discussion of the main findings and future research directions.

\section{Model and Notation}\label{sec2}

Consider longitudinal or clustered observations indexed by subject
\(i=1,\ldots,m\) and repeated measurements \(j=1,\ldots,n_i\).
Let \(Y_{ij}\) denote the response variable for the \(j\)th observation
within subject \(i\). Let
\(\mathbf{x}_{ij}=(x_{ij1},\ldots,x_{ijp})^T\) denote the vector of
population-level predictors and
\(\mathbf{z}_{ij}=(z_{ij1},\ldots,z_{ijq})^T\) denote the vector of
random-effect covariates. We assume that observations within the same subject may be correlated through subject-specific heterogeneity. Classical generalized linear mixed models (GLMMs) model the conditional
mean of the response as
\citep{breslow1993approximate,mcculloch2008generalized}
\begin{align}\label{GLMM}
g\{E(Y_{ij}\mid \mathbf{b}_i)\}
=
\mathbf{x}_{ij}^{T}\boldsymbol{\beta}
+
\mathbf{z}_{ij}^{T}\mathbf{b}_i,
\end{align}
where \(g(\cdot)\) is a known link function,
\(\boldsymbol{\beta}\) is the vector of fixed-effect coefficients, and
\(\mathbf{b}_i\) is the subject-specific random effect. Typically,
$\mathbf{b}_i \sim MVN(\mathbf{0},\mathbf{D})$,
where \(\mathbf{D}\) is the covariance matrix of the random effects. GLMMs provide a flexible framework for modeling within-subject dependence through subject-specific random effects. However, the population-level predictor
\(\mathbf{x}_{ij}^{T}\boldsymbol{\beta}\) imposes a linear structure on
the covariate effects. When the underlying response surface is strongly
nonlinear or involves complex interactions among predictors, this
linear specification may lead to inadequate prediction. On the other hand, standard artificial neural networks (ANNs) provide
flexible approximations to nonlinear predictor--response relationships but generally do not explicitly represent the subject-specific heterogeneity arising in longitudinal data
\citep{hornik1989multilayer,goodfellow2016deep}. The ANN model can be
expressed as
\begin{align}
g\{E(Y_{ij})\}
=
f_{\theta}(\mathbf{x}_{ij}),
\end{align}
where \(f_{\theta}(\cdot)\) denotes a feed-forward neural network.
To account for subject-specific heterogeneity,
\citet{maity2013repeated} proposed the Maity--Pal (MP) model, which
augments the neural network input with subject indicator variables. Let
\[
\tilde{\mathbf{x}}_{ij}
=
(\mathbf{x}_{ij}^{T},\mathbf{x}_{ijI}^{T})^{T},
\]
where \(\mathbf{x}_{ijI}\) is an \(m\)-dimensional indicator vector
identifying subject \(i\). The MP model is given by
\begin{align}\label{MP}
g\{E(Y_{ij})\}
=
f_{\theta}(\tilde{\mathbf{x}}_{ij}).
\end{align}
The MP model incorporates subject-specific information directly through
the network inputs. However, the dimension of the indicator vector
increases with the number of subjects, which may reduce computational efficiency and increase the risk of overfitting in larger longitudinal studies. To combine flexible nonlinear modeling with a conventional random-effects structure, \citet{mandel2023neural} proposed the
Generalized Neural Network Mixed Model (GNMM), defined as
\begin{align}
g\{E(Y_{ij}\mid \mathbf{b}_i)\}
=
f_{\theta}(\mathbf{x}_{ij})
+
\mathbf{z}_{ij}^{T}\mathbf{b}_i.
\end{align}
The GNMM replaces the linear population-level component of the GLMM
with a neural network while retaining the random-effects structure.
Its general formulation can accommodate multilayer feed-forward
networks, although the empirical investigation of
\citet{mandel2023neural} focused primarily on shallow one- and
two-hidden-layer networks. Building on this framework, we consider
increased nonlinear representation capacity together with
subject-specific heterogeneity for longitudinal prediction.

Specifically, we consider a Deep Generalized Neural Network Mixed Model
(DGNMM), in which the population-level component is represented by a
multilayer feed-forward neural network,
\begin{align}\label{DGNMM}
g\{E(Y_{ij}\mid \mathbf{b}_i)\}
=
f_{\theta}^{Deep}(\mathbf{x}_{ij})
+
\mathbf{z}_{ij}^{T}\mathbf{b}_i,
\end{align}
where \(f_{\theta}^{Deep}(\cdot)\) denotes a multilayer feed-forward
neural network parameterized by \(\theta\). The deep component flexibly approximates an unknown population-level response surface, while the subject-specific component accounts for heterogeneity among repeated measurements from different subjects. For simplicity, we focus on a random-intercept structure. The corresponding mixed-model representation is
\begin{align}\label{inDGNMM}
g\{E(Y_{ij}\mid b_i)\}
=
f_{\theta}^{Deep}(\mathbf{x}_{ij})
+
b_i,
\end{align}
where
\[
b_i\sim N(0,\tau^2).
\]
Equation~\eqref{inDGNMM} provides the conventional random-intercept
representation motivating the subject-specific component. For the
prediction-oriented estimation procedure developed in this work, we represent the subject-specific deviation by a directly trainable,
regularized scalar latent effect \(e_i\). The resulting working DGNMM is
\begin{align}\label{embeddingDGNMM}
g\{E(Y_{ij}\mid e_i)\}
=
f_{\theta}^{Deep}(\mathbf{x}_{ij})
+
e_i.
\end{align}
Here, \(e_i\) is shared across repeated observations from subject \(i\)
and captures persistent subject-specific deviation from the
population-level nonlinear response surface. Unlike the conventional
random effect \(b_i\), \(e_i\) is estimated directly and jointly with
the neural-network parameters rather than integrated over a prespecified
distribution. Regularization toward zero controls unrestricted
subject-specific adaptation and, as developed in the following section,
provides a connection to Gaussian random-intercept shrinkage.
The decomposition in \eqref{embeddingDGNMM} can be expressed as
\[
\underbrace{f_{\theta}^{Deep}(\mathbf{x}_{ij})}
_{\text{nonlinear population-level component}}
+
\underbrace{e_i}
_{\text{subject-specific deviation}}.
\]
Thus, the proposed formulation separates flexible estimation of the population-level nonlinear response surface from subject-specific heterogeneity. The DGNMM implementation considered in this study uses three hidden
layers with architecture
\[
p \rightarrow 64 \rightarrow 32 \rightarrow 16 \rightarrow 1,
\]
with ReLU activation functions in all hidden layers
\citep{nair2010rectified}, where \(p\) is the number of input
predictors. The hidden layers contain 64, 32, and 16 nodes,
respectively, followed by an output layer. Let
\begin{align}
\mathbf{h}_{ij}^{(1)}
=
\phi
\left(
\mathbf{W}^{(1)}
\mathbf{x}_{ij}
+
\mathbf{a}^{(1)}
\right)
\end{align}
denote the first hidden layer, where \(\phi(\cdot)\) represents the
ReLU activation function, \(\mathbf{W}^{(1)}\) is the weight matrix,
and \(\mathbf{a}^{(1)}\) is the bias vector. Similarly, deeper hidden
layers are recursively defined as
\begin{align}
\mathbf{h}_{ij}^{(\ell)}
=
\phi
\left(
\mathbf{W}^{(\ell)}
\mathbf{h}_{ij}^{(\ell-1)}
+
\mathbf{a}^{(\ell)}
\right),
\qquad
\ell=2,\ldots,L,
\end{align}
where \(L\) denotes the number of hidden layers. The final neural
network output is given by
\begin{align}\label{depp}
f_{\theta}^{Deep}(\mathbf{x}_{ij})
=
\mathbf{W}^{(L+1)}
\mathbf{h}_{ij}^{(L)}
+
\mathbf{a}^{(L+1)},
\end{align}
where
\begin{align}\label{theta1}
\theta
=
\left\{
\mathbf{W}^{(1)},\ldots,\mathbf{W}^{(L+1)},
\mathbf{a}^{(1)},\ldots,\mathbf{a}^{(L+1)}
\right\}.
\end{align}

\begin{figure}
\centering
\includegraphics[width=0.95\linewidth]{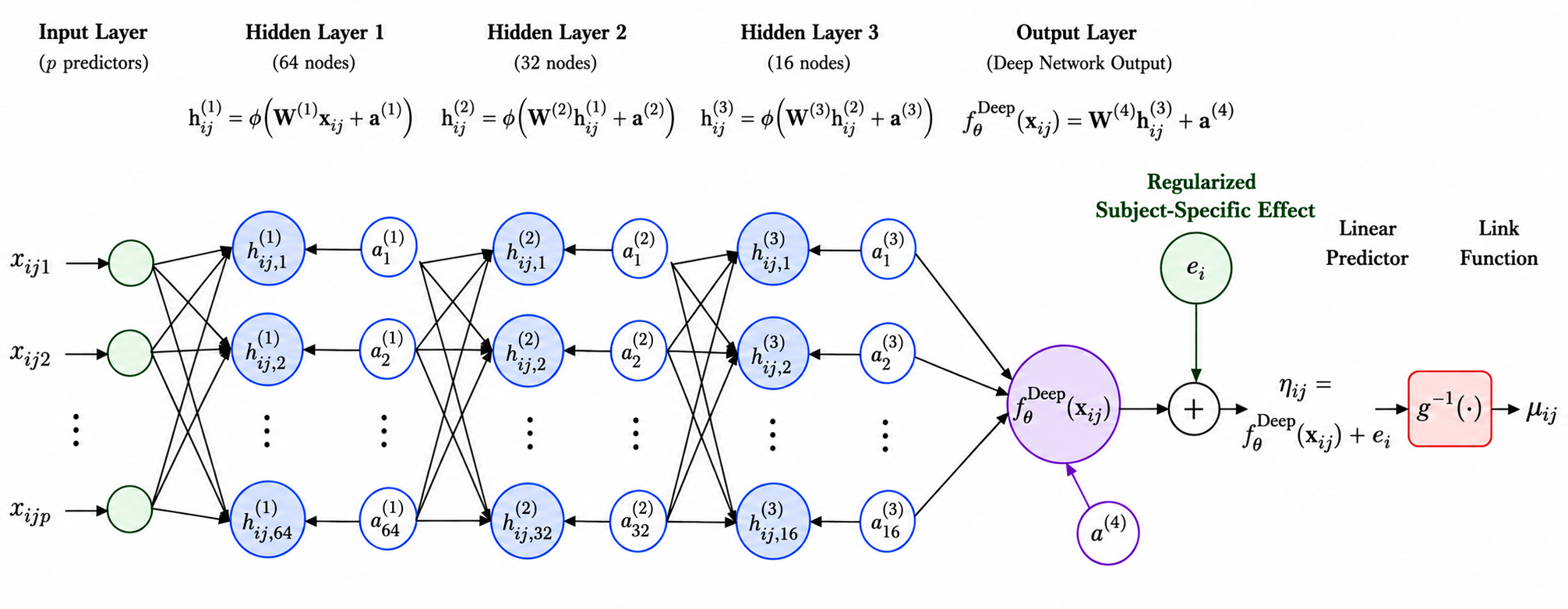}
\caption{
\textbf{Architecture of the proposed Deep Generalized Neural Network
Mixed Model (DGNMM).} The population-level nonlinear component
$f_{\theta}^{Deep}(\mathbf{x}_{ij})$ is represented by a multilayer
feed-forward neural network, while a regularized subject-specific latent
effect $e_i$ captures persistent individual deviation from the
population-level response surface. The resulting predictor is
$\eta_{ij}=f_{\theta}^{Deep}(\mathbf{x}_{ij})+e_i$, with the conditional
mean obtained through the inverse link function. Both components are
estimated jointly for longitudinal prediction.
}
\label{Figure:dgnmm}
\end{figure}

The parameter estimation procedure, regularization of the
subject-specific effects, prediction methodology, and performance
evaluation criteria for the proposed DGNMM are presented in the
following section.
\section{Estimation and Prediction Performance}\label{sec3}

Likelihood-based estimation in GLMMs and quasi-likelihood-based
estimation in neural mixed models provide principled approaches for
estimating fixed and random effects. However, as the nonlinear
population-level component becomes increasingly flexible, repeated
likelihood evaluation and integration over random effects can become
computationally demanding. Because the primary objective of the
proposed DGNMM is nonlinear longitudinal prediction, we adopt a
prediction-oriented estimation strategy in which the deep neural
network parameters and regularized subject-specific effects are learned
jointly through stochastic gradient-based optimization.

Specifically, for subject \(i\), the subject-specific contribution is
represented computationally by the trainable effect \(e_i\) introduced
in the preceding section. All repeated observations from subject \(i\)
share the same \(e_i\), which captures persistent subject-specific
heterogeneity not explained by the population-level nonlinear response
surface \(f_{\theta}^{Deep}(\mathbf{x}_{ij})\). Unlike a conventional
random effect, \(e_i\) is estimated directly together with the neural
network parameters rather than integrated over a prespecified
distribution.

For continuous outcomes, we assume
\begin{align}
Y_{ij}
=
f_{\theta}^{Deep}(\mathbf{x}_{ij})
+
e_i
+
\epsilon_{ij},
\end{align}
where
\[
\epsilon_{ij}\sim N(0,\sigma^2).
\]
For binary outcomes, we use the logistic link
\begin{align}
\mathrm{logit}\{P(Y_{ij}=1\mid e_i)\}
=
f_{\theta}^{Deep}(\mathbf{x}_{ij})
+
e_i.
\end{align}

The neural network parameters defined in the preceding section and the
subject-specific effects \(e_i\) are estimated jointly by minimizing a
regularized prediction loss using stochastic gradient-based optimization
with backpropagation and the Adam algorithm \citep{kingma2015adam}. In
addition to regularization of the neural network weights, we impose an
\(L_2\) penalty on the subject-specific effects. The latter shrinks
subject-specific deviations toward the population-level response
surface and prevents subjects with relatively few repeated
measurements from receiving excessively large effects.

For continuous outcomes, the proposed objective function is
\begin{align}\label{conloss}
\mathcal{L}_{C}
=&
\frac{1}{N}
\sum_{i=1}^{m}
\sum_{j=1}^{n_i}
\left(
Y_{ij}
-
\hat{Y}_{ij}
\right)^2
+
\lambda_{W}
\sum_{\ell=1}^{L+1}
\left\|
\mathbf{W}^{(\ell)}
\right\|_2^2
+
\lambda_{e}
\sum_{i=1}^{m}e_i^2,
\end{align}
where
\[
N=\sum_{i=1}^{m}n_i,
\]
\(\hat{Y}_{ij}\) denotes the predicted response,
\(\lambda_W>0\) controls regularization of the neural network weights,
and \(\lambda_e>0\) controls shrinkage of the subject-specific effects.

For binary outcomes, the corresponding regularized binary
cross-entropy objective is
\begin{align}\label{BinaryLoss}
\mathcal{L}_{B}
=&
-
\frac{1}{N}
\sum_{i=1}^{m}
\sum_{j=1}^{n_i}
\left[
Y_{ij}\log(\hat{\mu}_{ij})
+
(1-Y_{ij})
\log(1-\hat{\mu}_{ij})
\right]
\nonumber\\
&+
\lambda_{W}
\sum_{\ell=1}^{L+1}
\left\|
\mathbf{W}^{(\ell)}
\right\|_2^2
+
\lambda_{e}
\sum_{i=1}^{m}e_i^2.
\end{align}

The two regularization parameters play different roles.
The parameter \(\lambda_W\) controls the complexity of the nonlinear
population-level response surface, whereas \(\lambda_e\) determines the
degree of shrinkage of subject-specific deviations toward zero. Thus,
large values of \(\lambda_e\) favor population-level prediction,
whereas smaller values permit stronger subject-specific adaptation.

The penalty on the subject-specific effects also provides a direct
connection between the prediction-oriented formulation and conventional
random-intercept models. This connection is summarized in the following
proposition.

\begin{proposition}
Let us consider the Gaussian longitudinal model
$Y_{ij}=\mathbf{x}_{ij}^{T}\boldsymbol{\beta}+e_i+\epsilon_{ij}$,
$i=1,\ldots,m$, $j=1,\ldots,n_i$, where
$\epsilon_{ij}\sim N(0,\sigma^2)$ independently and
$e_i\sim N(0,\tau^2)$. Then, conditional on $\sigma^2$ and $\tau^2$,
maximizing the joint density with respect to
$\boldsymbol{\beta}$ and
$\mathbf{e}=(e_1,\ldots,e_m)^T$ is equivalent to minimizing
\begin{align}
Q(\boldsymbol{\beta},\mathbf{e})
&=
\sum_{i=1}^{m}\sum_{j=1}^{n_i}
\left[
Y_{ij}-\mathbf{x}_{ij}^{T}\boldsymbol{\beta}-e_i
\right]^2
+
\frac{\sigma^2}{\tau^2}\sum_{i=1}^{m}e_i^2.
\label{eq:penalized_Q}
\end{align}
\end{proposition}

\begin{proof}
Under the stated model, the conditional distribution of $Y_{ij}$ given
$e_i$ is
$Y_{ij}\mid e_i\sim
N(\mathbf{x}_{ij}^{T}\boldsymbol{\beta}+e_i,\sigma^2)$.
Hence,
\[
p(Y_{ij}\mid e_i,\boldsymbol{\beta})
=
\frac{1}{\sqrt{2\pi\sigma^2}}
\exp\left\{
-\frac{1}{2\sigma^2}
\left[
Y_{ij}-\mathbf{x}_{ij}^{T}\boldsymbol{\beta}-e_i
\right]^2
\right\}.
\]

Conditional on the subject-specific effects, the observations are
independent. Therefore, the conditional likelihood is
\begin{align*}
L(\boldsymbol{\beta},\mathbf{e}\mid\mathbf{Y})
&=
\prod_{i=1}^{m}\prod_{j=1}^{n_i}
p(Y_{ij}\mid e_i,\boldsymbol{\beta})\\
&=
(2\pi\sigma^2)^{-N/2}
\exp\left\{
-\frac{1}{2\sigma^2}
\sum_{i=1}^{m}\sum_{j=1}^{n_i}
\left[
Y_{ij}-\mathbf{x}_{ij}^{T}\boldsymbol{\beta}-e_i
\right]^2
\right\},
\end{align*}
where $N=\sum_{i=1}^{m}n_i$.

Since $e_i\sim N(0,\tau^2)$, its density is
\[
p(e_i)
=
\frac{1}{\sqrt{2\pi\tau^2}}
\exp\left\{-\frac{e_i^2}{2\tau^2}\right\}.
\]
Assuming independence across subjects, the joint density of
$\mathbf{e}$ is
\[
p(\mathbf{e})
=
\prod_{i=1}^{m}p(e_i)
=
(2\pi\tau^2)^{-m/2}
\exp\left\{
-\frac{1}{2\tau^2}\sum_{i=1}^{m}e_i^2
\right\}.
\]

Combining the conditional likelihood and the Gaussian density of the
subject-specific effects gives
\[
p(\mathbf{Y},\mathbf{e}\mid\boldsymbol{\beta})
\propto
\exp\left\{
-\frac{1}{2\sigma^2}
\sum_{i=1}^{m}\sum_{j=1}^{n_i}
\left[
Y_{ij}-\mathbf{x}_{ij}^{T}\boldsymbol{\beta}-e_i
\right]^2
-\frac{1}{2\tau^2}\sum_{i=1}^{m}e_i^2
\right\}.
\]

Taking the negative logarithm and omitting terms that do not depend on
$\boldsymbol{\beta}$ or $\mathbf{e}$ yields
\[
-\log p(\mathbf{Y},\mathbf{e}\mid\boldsymbol{\beta})
=
\frac{1}{2\sigma^2}
\sum_{i=1}^{m}\sum_{j=1}^{n_i}
\left[
Y_{ij}-\mathbf{x}_{ij}^{T}\boldsymbol{\beta}-e_i
\right]^2
+
\frac{1}{2\tau^2}\sum_{i=1}^{m}e_i^2
+C,
\]
where $C$ is constant with respect to
$\boldsymbol{\beta}$ and $\mathbf{e}$.

Thus, maximizing the joint density is equivalent to minimizing
\[
\frac{1}{2\sigma^2}
\sum_{i=1}^{m}\sum_{j=1}^{n_i}
\left[
Y_{ij}-\mathbf{x}_{ij}^{T}\boldsymbol{\beta}-e_i
\right]^2
+
\frac{1}{2\tau^2}\sum_{i=1}^{m}e_i^2.
\]
Multiplying by the positive constant $2\sigma^2$ does not change the
minimizer and gives
\[
Q(\boldsymbol{\beta},\mathbf{e})
=
\sum_{i=1}^{m}\sum_{j=1}^{n_i}
\left[
Y_{ij}-\mathbf{x}_{ij}^{T}\boldsymbol{\beta}-e_i
\right]^2
+
\frac{\sigma^2}{\tau^2}\sum_{i=1}^{m}e_i^2,
\]
which proves the result.
\end{proof}
\medskip

\noindent
\begin{remark}
Proposition~1 provides a statistical interpretation of the
regularization imposed on the subject-specific effects in DGNMM. In the
Gaussian linear special case, the random-intercept formulation induces
a quadratic penalty proportional to
\[
\frac{\sigma^2}{\tau^2}\sum_{i=1}^{m}e_i^2.
\]
Thus, the \(L_2\) penalty used in DGNMM has the same quadratic shrinkage
form as that induced by a Gaussian random intercept. Because the
prediction loss in \eqref{conloss} is expressed as an average over
\(N\) observations, the corresponding coefficient under this
normalization is \(\sigma^2/(N\tau^2)\). In the proposed DGNMM,
\(\lambda_e\) is treated as a tuning parameter and selected using
validation prediction performance rather than estimated through
variance components. Larger values of \(\lambda_e\) produce stronger
shrinkage toward zero, whereas smaller values permit greater
subject-specific adaptation. The DGNMM extends this shrinkage principle
to nonlinear longitudinal prediction by replacing the linear
population-level component with
\(f_{\theta}^{Deep}(\mathbf{x}_{ij})\).
\end{remark}
\medskip

The regularization parameters are selected without using the holdout
observations. In the numerical studies, $\lambda_W$ and $\lambda_e$ are
selected from prespecified candidate values using validation prediction
performance, after which the selected values are used for the final
model fit. This ensures that the holdout observations are used only for the final evaluation of predictive performance.

\begin{algorithm}[ht]
\caption{Parameter Estimation Procedure for DGNMM}
\label{alg:dgnmm}
\begin{enumerate}

\item Initialize the neural network parameters \(\theta\) and
subject-specific effects \(e_i\), \(i=1,\ldots,m\).

\item For each stochastic update, compute the linear predictor for the
current observation,
\[
\eta_{ij}
=
f_{\theta}^{Deep}(\mathbf{x}_{ij})
+
e_i.
\]

\item Compute the predicted response \(\hat{Y}_{ij}\) for continuous
outcomes or the predicted probability \(\hat{\mu}_{ij}\) for binary
outcomes.

\item Evaluate the corresponding regularized objective function
defined in \eqref{conloss} or \eqref{BinaryLoss}, including both the
network-weight penalty and the subject-specific effect penalty.

\item Compute gradients of the regularized objective function with
respect to \(\theta\) and \(e_i\) using backpropagation.

\item Update \(\theta\) and \(e_i\) using the Adam optimization
algorithm.

\item Repeat Steps 2--6 over successive observations and training
epochs until the convergence criterion is satisfied.

\item Obtain the final estimates \(\hat{\theta}\) and
\(\hat{e}_i\), and compute longitudinal predictions using the fitted
model.

\end{enumerate}
\end{algorithm}
Let \(\hat{\theta}\) and \(\hat{e}_i\) denote the resulting estimates
after convergence. The fitted linear predictor for an existing subject
is
\[
\widehat{\eta}_{ij}
=
f_{\hat{\theta}}^{Deep}(\mathbf{x}_{ij})
+
\hat{e}_i.
\]
For continuous outcomes, the predicted response is
\[
\hat{Y}_{ij}
=
\widehat{\eta}_{ij}.
\]
For binary outcomes, the predicted probability is
\[
\hat{\mu}_{ij}
=
\frac{\exp(\widehat{\eta}_{ij})}
{1+\exp(\widehat{\eta}_{ij})}.
\]

The primary prediction task considered in this paper is
within-subject longitudinal prediction. Specifically, the final
observation from each subject is reserved as a holdout observation,
while the preceding observations from that subject are available for
estimating the subject-specific effect. This design evaluates the
ability of the proposed model to predict a future response for a
subject with an observed longitudinal history.

For continuous outcomes, predictive accuracy is assessed using the
Mean Squared Prediction Error (MSPE) \citep{hastie2009elements},
defined as
\begin{align}\label{MSPE}
MSPE
=
\frac{1}{m}
\sum_{i=1}^{m}
\left(
Y_{i,n_i}
-
\hat{Y}_{i,n_i}
\right)^2,
\end{align}
where \(Y_{i,n_i}\) is the observed holdout outcome for subject \(i\)
and \(\hat{Y}_{i,n_i}\) is the corresponding predicted value. Smaller
MSPE values indicate better predictive performance.

For binary outcomes, predictive performance is assessed using the Area
Under the Receiver Operating Characteristic Curve (AUC)
\citep{hanley1982auc}. The AUC is calculated from the observed holdout
responses \(Y_{i,n_i}\) and the corresponding predicted probabilities
\(\hat{\mu}_{i,n_i}\). Larger AUC values indicate better discrimination,
with a value of one corresponding to perfect discrimination and a value
of 0.5 corresponding to random classification.

MSPE and AUC are used as the primary evaluation measures because they
directly reflect the two prediction tasks considered here: accuracy of
continuous-response prediction and discrimination for binary outcomes, respectively. Restricting evaluation to these task-specific measures
also provides a common and interpretable basis for comparing all
competing longitudinal prediction methods without introducing multiple redundant performance criteria. The simulation study in the following section examines how the predictive performance of DGNMM changes as the complexity of the
population-level response surface and the magnitude of subject-specific
heterogeneity are varied independently.
\section{Simulation Study}\label{sec4}
We conducted a simulation study to evaluate the predictive performance
of the proposed Deep Generalized Neural Network Mixed Model (DGNMM) for
longitudinal and clustered data. The simulation was designed to examine
how the predictive benefit of increased nonlinear representation
capacity changes with the degree of population-level nonlinearity and
the magnitude of subject-specific heterogeneity. The proposed method was
compared with generalized linear mixed models (GLMM), standard
artificial neural networks (ANN), neural networks with subject
indicators (MP), and generalized neural network mixed models with one
and two hidden layers (GNMM-1 and GNMM-2).

\subsection{Simulation Design}

Let $i=1,\ldots,m$ denote subjects and $j=1,\ldots,n_i$ denote repeated
observations within subject $i$. We set $m=100$, and the number of
observations for each subject was generated as
$n_i\sim\mathrm{Poisson}(6)+2$, ensuring that each subject had multiple
observations available for model fitting and prediction. For each
observation, six independent binary covariates were generated as
$x_{ij1},\ldots,x_{ij6}\sim\mathrm{Bernoulli}(0.5)$. A subject-specific
random intercept was generated as $b_i\sim N(0,\tau^2)$, where $\tau$
controls the magnitude of subject-specific heterogeneity. For continuous
outcomes, the random error term was generated as
$\epsilon_{ij}\sim N(0,0.0025)$.

To systematically vary the degree of population-level nonlinearity, we
defined the response surface as
$f_{\gamma}(\mathbf{x}_{ij})=(1-\gamma)f_L(\mathbf{x}_{ij})
+\gamma f_{NL}(\mathbf{x}_{ij})$, where
$\gamma\in\{0,0.25,0.50,0.75,1\}$. The parameter $\gamma$ controls the
degree of nonlinearity, with $\gamma=0$ corresponding to a completely
linear response surface and $\gamma=1$ corresponding to the fully
nonlinear response surface. The intermediate values represent gradual
transitions between the linear and nonlinear settings and allow the
predictive benefit of increased nonlinear representation capacity to be
examined directly.

The linear component was defined as
$f_L(\mathbf{x}_{ij})=\mathbf{x}_{ij}^{T}\boldsymbol{\beta}_L$. For
continuous outcomes,
$\boldsymbol{\beta}_L\sim MVN(0.5\mathbf{1}_6,I_6)$, whereas for binary
outcomes,
$\boldsymbol{\beta}_L\sim MVN(0.3\mathbf{1}_6,3I_6)$.

To construct the nonlinear component, three nonlinear indicator
variables were defined as
$Z_{ij1}=I(x_{ij1}+x_{ij2}=1)$,
$Z_{ij2}=I(x_{ij3}+x_{ij4}\neq1)$, and
$Z_{ij3}=I(x_{ij5}+x_{ij6}=1)$. Let
$\mathbf{Z}_{ij}=(Z_{ij1},Z_{ij2},Z_{ij3})^T$. The nonlinear component
was defined as
$f_{NL}(\mathbf{x}_{ij})=\mathbf{Z}_{ij}^{T}
\boldsymbol{\beta}_{NL}$. For continuous outcomes,
$\boldsymbol{\beta}_{NL}\sim MVN(2\mathbf{1}_3,I_3)$, whereas for binary
outcomes,
$\boldsymbol{\beta}_{NL}\sim MVN(3\mathbf{1}_3,I_3)$.
For each value of $\gamma$, continuous outcomes were generated as
$Y_{ij}=f_{\gamma}(\mathbf{x}_{ij})+b_i+\epsilon_{ij}$. For binary
outcomes, data were generated according to
$\mathrm{logit}(\mu_{ij})=f_{\gamma}(\mathbf{x}_{ij})+b_i$ and
$Y_{ij}\sim\mathrm{Bernoulli}(\mu_{ij})$.

Subject-specific heterogeneity was varied independently of the degree of
population-level nonlinearity. We considered
$\tau\in\{0,0.25,0.50,1.00\}$, where $\tau=0$ represents the absence of
subject-specific heterogeneity and larger values correspond to
increasing between-subject variation. Crossing the five values of
$\gamma$ with the four values of $\tau$ resulted in 20 simulation
settings for each outcome type. This design allows the effects of
population-level nonlinearity and subject-specific heterogeneity on
longitudinal prediction to be examined separately and jointly.

Network parameters in the ANN, MP, GNMM, and DGNMM models were estimated using stochastic gradient-based optimization with backpropagation and the Adam algorithm. For the GNMM competitors, two network architectures were considered. GNMM-1 contained one hidden layer, whereas GNMM-2 contained two hidden layers. Two sizes, five and ten
nodes, were considered for the first hidden layer, with five nodes in
the second hidden layer for GNMM-2. The proposed DGNMM employed a deep
neural network architecture
$6\rightarrow64\rightarrow32\rightarrow16\rightarrow1$, with ReLU
activation functions in all hidden layers.
The DGNMM was fitted with separate $L_2$ regularization terms for the neural network weights and subject-specific effects, governed by $\lambda_W$ and $\lambda_e$, respectively. The tuning parameters were selected using prediction performance on validation observations
formed from the training data, without using the final holdout
observations.

For each simulated dataset, the final observation from each subject was
reserved as a holdout test observation, while all preceding observations
from the same subject were used for model fitting and estimation of the
subject-specific effect. Thus, the simulation focuses on within-subject
longitudinal prediction, where historical observations from a subject
are available for predicting a future response. Predictive performance
was assessed using MSPE for continuous outcomes and AUC for binary
outcomes.

All simulation results were based on 500 independent replications for
each combination of $\gamma$ and $\tau$. For each replication, the
final observation from each subject was used only for evaluating
prediction performance, and the resulting MSPE or AUC was summarized
across the 500 replications.

The simulation study was designed to address three primary questions:
(i) whether increased nonlinear representation capacity improves
prediction as the population-level response surface becomes increasingly
nonlinear; (ii) how relative predictive performance changes with the magnitude of subject-specific heterogeneity; and (iii) whether the regularized DGNMM remains competitive when the underlying response structure is
linear or only mildly nonlinear.

\subsection{Results Discussion}
Figure~\ref{fig:MSPE5} presents the continuous-outcome MSPE results when
5 nodes are used in the first hidden layer of the competing shallow
neural-network models. Across the 20 combinations of $\gamma$ and
$\tau$, DGNMM achieves the lowest MSPE in 15 settings, whereas GLMM is
best in the remaining 5 settings. The pattern of these wins is highly
structured. When $\gamma=0$, GLMM is the best-performing method for all
four values of $\tau$, reflecting correct specification of the linear
population-level response. At $\gamma=0.25$, DGNMM becomes the
best-performing method for $\tau=0$, $0.25$, and $0.50$, while GLMM
remains best when $\tau=1$. Once the degree of nonlinearity reaches
$\gamma=0.50$, DGNMM has the lowest MSPE for every considered value of
$\tau$, and this dominance continues at $\gamma=0.75$ and
$\gamma=1.00$.
\begin{figure}[ht]
\centering
\includegraphics[width=0.75\linewidth]
{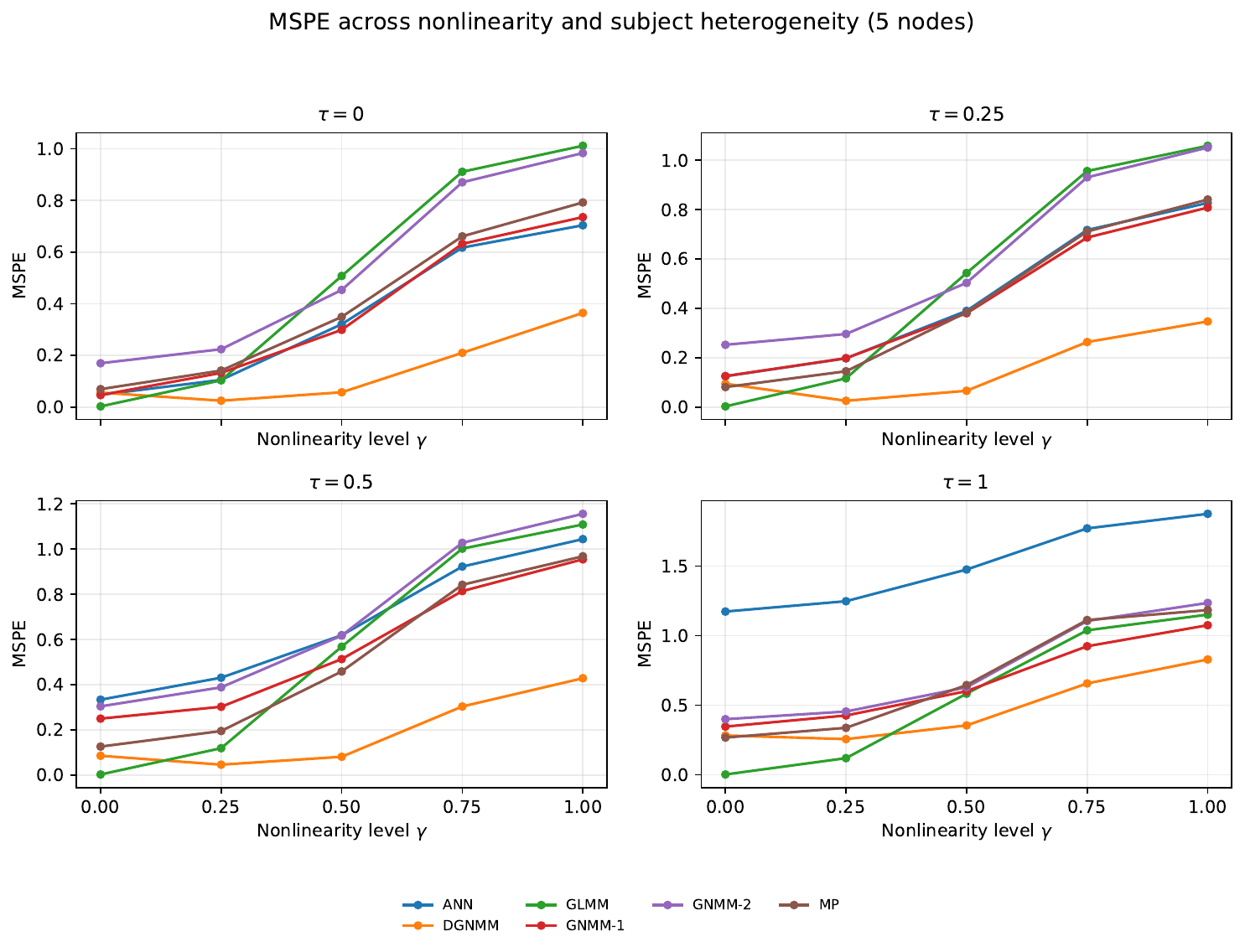}
\caption{
MSPE comparison of DGNMM, GNMM-1, GNMM-2, ANN, GLMM, and MP across
different levels of population-level nonlinearity $\gamma$ and
subject-specific heterogeneity $\tau$ for continuous outcomes. Results
are presented for competing neural-network specifications with 5 nodes
in the first hidden layer. The proposed DGNMM employs a deep neural
network architecture with hidden layers of sizes 64, 32, and 16.
For each combination of $\gamma$ and $\tau$, the final observation from
each subject was reserved as a holdout test sample, and prediction
performance was evaluated using MSPE over 500 simulation replications.
Lower MSPE indicates better predictive performance.
}
\label{fig:MSPE5}
\end{figure}
The magnitude of the differences becomes substantial under stronger
nonlinearity. For example, at $\gamma=1$ and $\tau=0$, the MSPE is
approximately 0.36 for DGNMM compared with approximately 1.01 for
GLMM. At $\tau=0.50$, the corresponding values are approximately 0.43
for DGNMM and 1.10 for GLMM. Even at the largest heterogeneity level,
$\tau=1$, DGNMM retains the lowest MSPE at $\gamma=1$, with an MSPE of
approximately 0.83 compared with approximately 1.15 for GLMM. These
results indicate a clear transition in the preferred model: GLMM is
strongest under complete linearity, whereas DGNMM becomes preferable as
the population-level response departs from linearity.
\begin{figure}[ht]
\centering
\includegraphics[width=0.75\linewidth]
{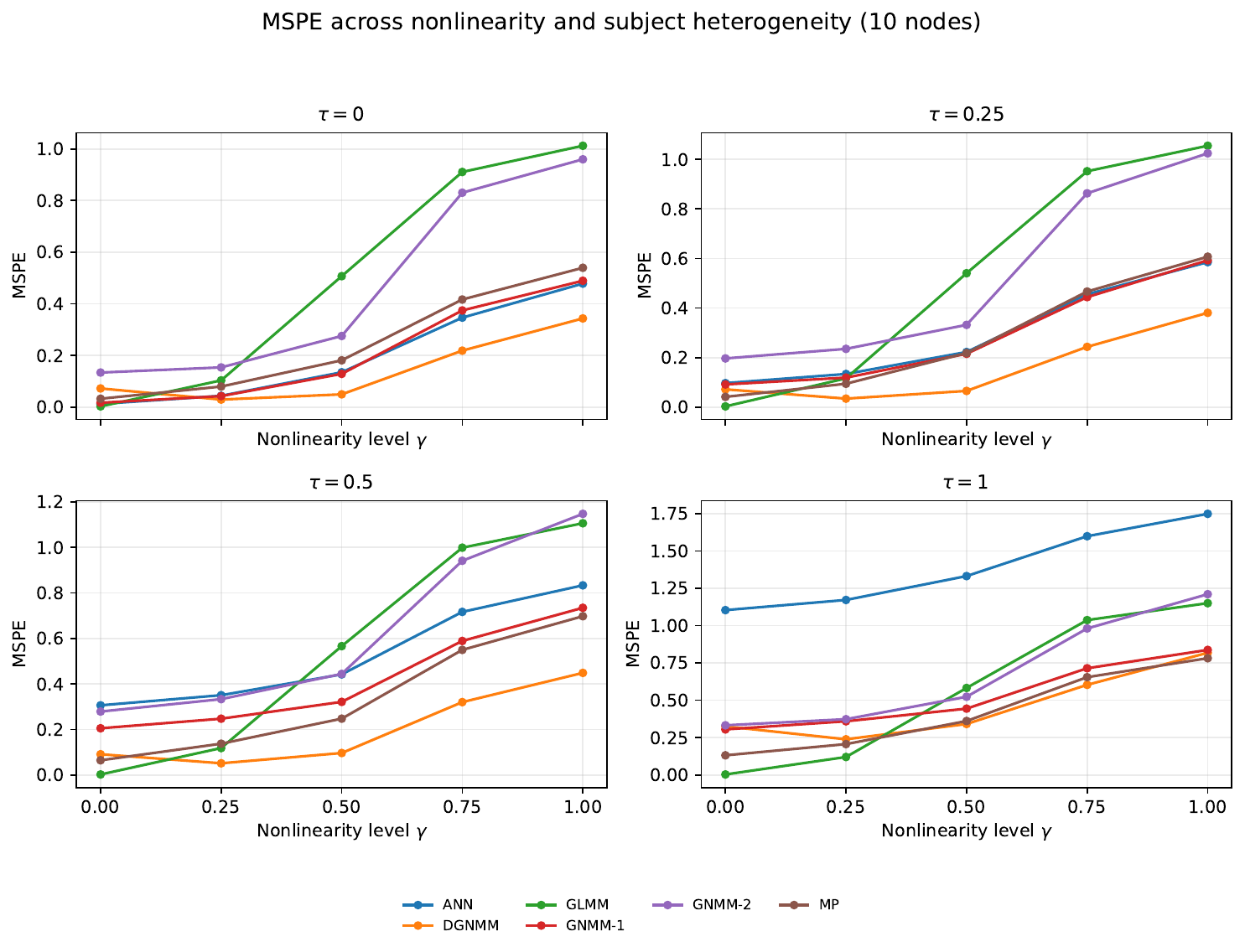}
\caption{
MSPE comparison of DGNMM, GNMM-1, GNMM-2, ANN, GLMM, and MP across
different levels of population-level nonlinearity $\gamma$ and
subject-specific heterogeneity $\tau$ for continuous outcomes. Results
are presented for competing neural-network specifications with 10 nodes
in the first hidden layer. The proposed DGNMM employs a deep neural
network architecture with hidden layers of sizes 64, 32, and 16.
For each combination of $\gamma$ and $\tau$, the final observation from
each subject was reserved as a holdout test sample, and prediction
performance was evaluated using MSPE over 500 simulation replications.
Lower MSPE indicates better predictive performance.
}
\label{fig:MSPE10}
\end{figure}
Figure~\ref{fig:MSPE10} shows that the same general pattern persists
when the first hidden layer of the competing shallow neural-network
models is increased to 10 nodes. DGNMM achieves the lowest MSPE in
14 of the 20 settings, GLMM in 5 settings, and MP in 1 setting. As in
the 5-node comparison, GLMM is best for all four values of $\tau$ when
$\gamma=0$. At $\gamma=0.25$, DGNMM is best for $\tau=0$, $0.25$,
and $0.50$, whereas GLMM remains best at $\tau=1$. DGNMM then achieves
the lowest MSPE for all four values of $\tau$ at $\gamma=0.50$ and
$\gamma=0.75$.

At $\gamma=1$, DGNMM remains best for $\tau=0$, $0.25$, and $0.50$.
For example, at $\tau=0$, the MSPE is approximately 0.34 for DGNMM
compared with approximately 1.01 for GLMM, while at $\tau=0.50$ the
corresponding values are approximately 0.45 and 1.10. The only setting
in which another nonlinear competitor becomes best is
$\gamma=1,\tau=1$, where MP has a slightly lower MSPE than DGNMM.
Thus, increasing the width of the shallower networks does not remove
the principal pattern observed in Figure~\ref{fig:MSPE5}: DGNMM is
most advantageous when moderate or strong nonlinear population-level
structure is present, but it is not uniformly optimal in every
high-heterogeneity setting. Comparison of Figures~\ref{fig:MSPE5} and~\ref{fig:MSPE10} also
clarifies the role of subject-specific heterogeneity. The predictive
advantage of DGNMM does not increase monotonically with $\tau$. Under
mild nonlinearity, for example, DGNMM is preferred at
$\tau\leq0.50$, but GLMM becomes preferable at $\tau=1$. In contrast,
under moderate and strong nonlinearity, DGNMM remains the best or nearly
best method over most heterogeneity levels. The continuous-response
results therefore indicate that nonlinear population-level complexity
is the primary setting in which increased representation capacity
provides a clear predictive benefit, while the magnitude of
subject-specific heterogeneity modifies, rather than uniformly
strengthens, that benefit.

\begin{figure}[ht]
\centering
\includegraphics[width=0.75\linewidth]
{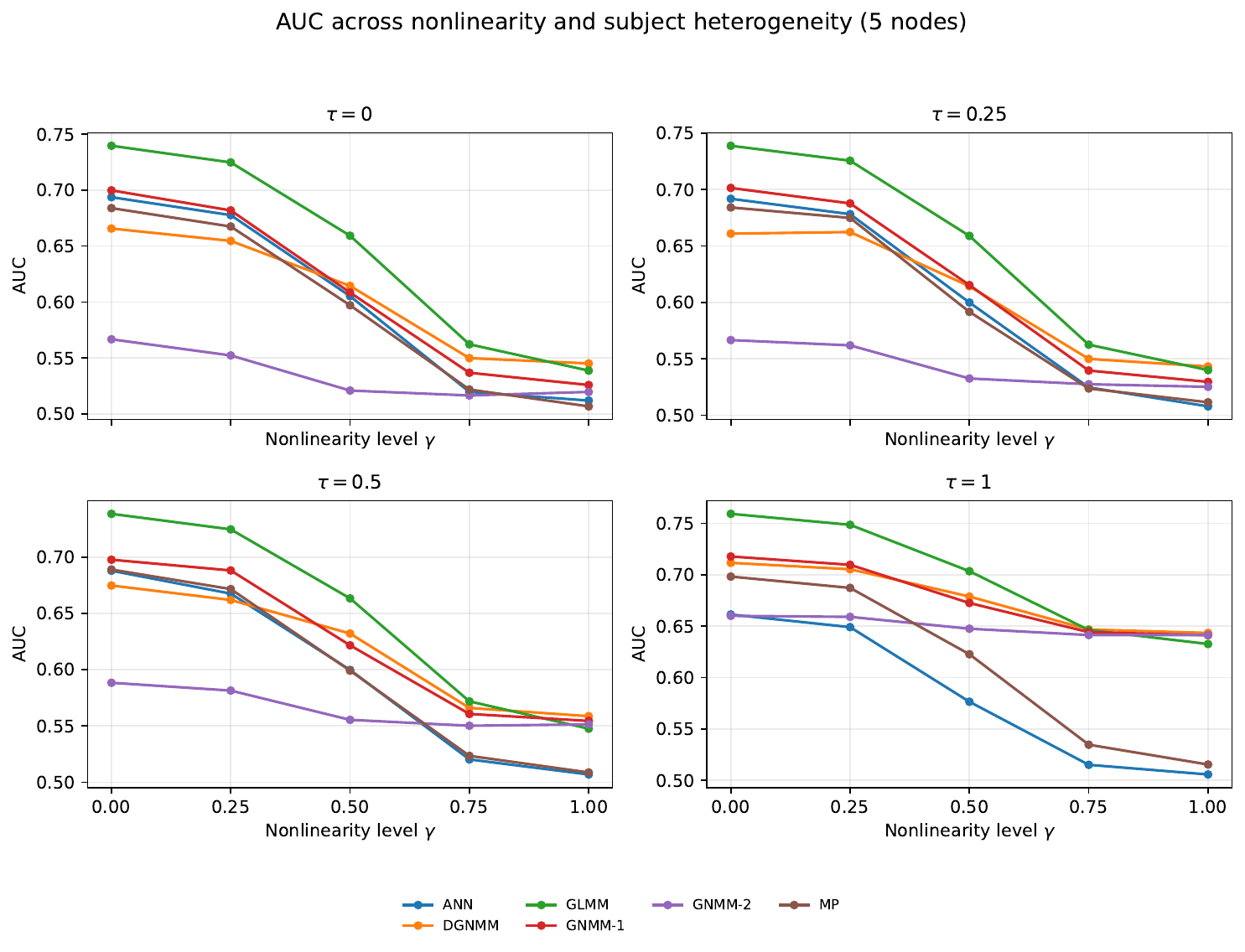}
\caption{
AUC comparison of DGNMM, GNMM-1, GNMM-2, ANN, GLMM, and MP across
different levels of population-level nonlinearity $\gamma$ and
subject-specific heterogeneity $\tau$ for binary outcomes. Results are
presented for competing neural-network specifications with 5 nodes in
the first hidden layer. The proposed DGNMM employs a deep neural network
architecture with hidden layers of sizes 64, 32, and 16. For each
combination of $\gamma$ and $\tau$, the final observation from each
subject was reserved as a holdout test sample, and prediction performance
was evaluated using AUC over 500 simulation replications. Larger AUC
indicates better discrimination.
}
\label{fig:AUC5}
\end{figure}
Figure~\ref{fig:AUC5} presents the binary-outcome results for the
5-node comparison. The pattern is markedly different from that observed
for MSPE. GLMM achieves the highest AUC in 15 of the 20 settings,
whereas DGNMM is best in 5 settings. In particular, GLMM is the
best-performing method for all four values of $\tau$ when $\gamma=0$
and remains best across most settings at $\gamma=0.25$ and
$\gamma=0.50$. For example, when $\gamma=0$ and $\tau=0$, the AUC is
approximately 0.74 for GLMM compared with approximately 0.67 for
DGNMM. At $\tau=1$, the corresponding AUC values are approximately
0.76 and 0.72.

The relative advantage of DGNMM emerges only under stronger nonlinear
structure. In the 5-node comparison, DGNMM attains the highest AUC at
$\gamma=1$ for all four values of $\tau$, and it also becomes the
best-performing method at $\gamma=0.75,\tau=1$. Thus, all five DGNMM
wins occur in settings characterized by strong nonlinearity, with one
of them also requiring the largest level of subject-specific
heterogeneity. At $\gamma=1,\tau=1$, DGNMM achieves an AUC of
approximately 0.64, whereas ANN and MP fall to approximately 0.51 and
0.52, respectively. GLMM and the GNMM alternatives remain relatively
close, indicating that the binary advantage is considerably smaller
than the continuous-response MSPE advantage.
\begin{figure}[ht]
\centering
\includegraphics[width=0.75\linewidth]
{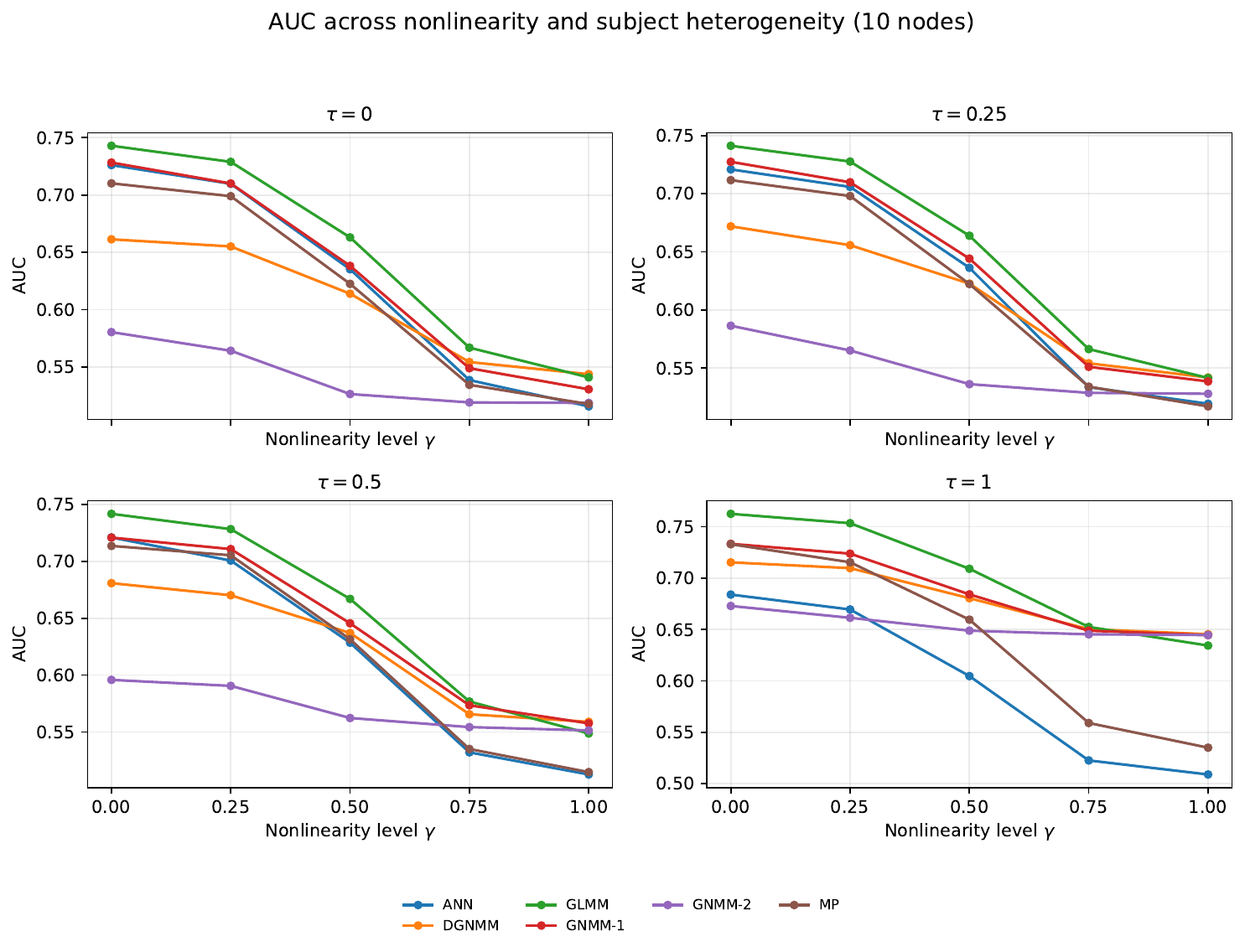}
\caption{
AUC comparison of DGNMM, GNMM-1, GNMM-2, ANN, GLMM, and MP across
different levels of population-level nonlinearity $\gamma$ and
subject-specific heterogeneity $\tau$ for binary outcomes. Results are
presented for competing neural-network specifications with 10 nodes in
the first hidden layer. The proposed DGNMM employs a deep neural network
architecture with hidden layers of sizes 64, 32, and 16. For each
combination of $\gamma$ and $\tau$, the final observation from each
subject was reserved as a holdout test sample, and prediction performance
was evaluated using AUC over 500 simulation replications. Larger AUC
indicates better discrimination.
}
\label{fig:AUC10}
\end{figure}
Figure~\ref{fig:AUC10} reinforces this conclusion. With 10 nodes in
the first hidden layer of the competing shallow networks, GLMM achieves
the highest AUC in 16 of the 20 settings, while DGNMM is best in the
remaining 4 settings. The four DGNMM wins occur at $\gamma=1$ across
the four values of $\tau$. Thus, unlike the continuous-outcome results,
increasing nonlinear representation capacity does not produce an
advantage at mild or moderate levels of nonlinearity. GLMM remains the
preferred method throughout $\gamma=0$, $0.25$, $0.50$, and $0.75$
across the considered levels of heterogeneity.

At the strongest nonlinear setting, however, DGNMM becomes the
best-performing method across all four values of $\tau$. The
differences are generally modest relative to the continuous-response
case. For instance, at $\gamma=1$ and $\tau=1$, DGNMM and GNMM-2 have
AUC values near 0.64, with GLMM and GNMM-1 also close to this level,
whereas ANN and MP fall to approximately 0.51 and 0.54, respectively.
Consequently, the binary results suggest that increased nonlinear
representation capacity becomes useful primarily when the departure
from linearity is sufficiently strong, but the improvement over
conventional and shallower mixed-model alternatives remains modest.

Taken together, the binary results show that outcome type strongly
affects the value of additional nonlinear representation capacity. In
the 5-node comparison, DGNMM is best in 5 of 20 settings, compared
with 15 of 20 for the continuous MSPE analysis. In the 10-node
comparison, the corresponding numbers are 4 of 20 for binary AUC and
14 of 20 for continuous MSPE. Hence, the advantage of DGNMM is much
more pronounced for continuous prediction than for binary
discrimination in the simulation settings considered here.

Overall, the simulations identify a clear boundary for the usefulness
of increased nonlinear representation capacity. For continuous
outcomes, GLMM is uniformly preferable under complete linearity, but
DGNMM becomes the best-performing method in most settings as soon as
moderate nonlinear structure is introduced and remains dominant across
the stronger nonlinear settings. For binary outcomes, GLMM remains
competitive over a much broader range of $\gamma$, and DGNMM becomes
best primarily at the strongest nonlinear setting. Moreover, the
predictive advantage of DGNMM does not increase monotonically with
$\tau$ in either outcome type. These findings indicate that the value
of the proposed model is determined primarily by the complexity of the
population-level response structure, with subject-specific
heterogeneity further modifying relative performance.

We next evaluate the proposed DGNMM using longitudinal smartphone
digital-phenotyping data for schizophrenia symptom prediction. The data
application and corresponding results are presented in
Section~\ref{sec5}.
\section{Real Data Application: Schizophrenia Symptom Prediction}
\label{sec5}

\subsection{Data Description}

We evaluated the predictive performance of the proposed DGNMM using a
longitudinal schizophrenia dataset previously analyzed by
\citep{torous2016new}. The study consisted of 15 individuals diagnosed
with schizophrenia who were monitored using the Beiwe smartphone-based
digital phenotyping platform. During a 90-day study period, GPS and
accelerometer measurements were continuously collected from participants'
smartphones, together with repeated self-reported assessments of depression
and anxiety symptoms. The number of observations per subject ranged from
2 to 121 days, with an average of 53 observations per individual.

We considered depression severity and anxiety severity as two longitudinal
prediction outcomes using the same 15 daily mobility-related features
derived from the GPS and accelerometer data. The continuous symptom scores
were used to evaluate prediction accuracy through MSPE. In a complementary
analysis, each outcome was dichotomized at its sample median to define
low- and high-symptom categories, and discrimination was evaluated using
AUC. Observations with scores above the median were classified as high
symptom severity, while observations with scores at or below the median
were classified as low symptom severity.

We compared DGNMM with GNMM-1, GNMM-2, GLMM, ANN, and MP. For DGNMM,
we used the architecture
$15\rightarrow64\rightarrow32\rightarrow16\rightarrow1$.
In addition to overall prediction performance, we further examined the
estimated subject-specific effects and assessed whether their magnitudes
were associated with subject-level predictive gains relative to the
main-effects GLMM.

\begin{table}[ht]
\centering
\caption{MSPE comparison of DGNMM, GNMM-1, GNMM-2, GLMM, ANN, and MP for depression and anxiety prediction. Lower MSPE indicates better predictive performance.}
\label{tab:schizophrenia_mspe}
\begin{tabular}{llc}
\hline
Outcome & Method & MSPE \\
\hline
Depression & DGNMM      & 0.324665 \\
           & GNMM-1     & 0.731966 \\
           & GNMM-2     & 0.748917 \\
           & ANN        & 0.986315 \\
           & GLMM Main Effects & 0.546328 \\
           & GLMM + 3 Interactions & 0.608985 \\
           & GLMM + 5 Interactions & 0.588429 \\
           & MP         & 0.632566 \\
\hline
Anxiety    & DGNMM      & 0.035504 \\
           & GNMM-1     & 0.300347 \\
           & GNMM-2     & 0.262234 \\
           & ANN        & 0.485459 \\
           & GLMM Main Effects & 0.068290 \\
           & GLMM + 3 Interactions & 0.097573 \\
           & GLMM + 5 Interactions & 0.116300 \\
           & MP         & 0.100010 \\
\hline
\end{tabular}
\end{table}

\subsection{Predictive Performance}

Table~\ref{tab:schizophrenia_mspe} presents the MSPE results for the
continuous depression and anxiety outcomes. DGNMM achieved the lowest
MSPE for both outcomes. For depression, DGNMM obtained an MSPE of
0.325, compared with 0.732 for GNMM-1, 0.749 for GNMM-2, 0.986 for ANN,
and 0.633 for MP. The main-effects GLMM obtained an MSPE of 0.546,
while the interaction-augmented GLMM specifications produced MSPE values
of 0.609 and 0.588.

For anxiety, DGNMM achieved an MSPE of 0.036, compared with 0.300 for
GNMM-1, 0.262 for GNMM-2, 0.485 for ANN, and 0.100 for MP. The
main-effects GLMM also performed relatively well, with an MSPE of 0.068,
but its prediction error remained larger than that of DGNMM.

These results show that DGNMM provided substantial predictive gains for
the continuous symptom outcomes in this application. For depression,
the MSPE of DGNMM was approximately 41\% lower than that of the
main-effects GLMM and approximately 56\% lower than that of GNMM-1.
For anxiety, DGNMM reduced MSPE by approximately 48\% relative to the
main-effects GLMM and by approximately 88\% relative to GNMM-1.
Thus, the improvement was observed relative to both conventional and
shallower neural mixed-model alternatives. However, these predictive
gains should not by themselves be interpreted as evidence of a strongly
nonlinear population-level mobility--symptom relationship, since the
DGNMM prediction combines the population-level neural representation
with regularized subject-specific effects.

\begin{table}[ht]
\centering
\caption{AUC comparison of DGNMM, GNMM-1, GNMM-2, GLMM, ANN, and MP for classifying high versus low depression and anxiety symptoms. Larger AUC values indicate better classification performance.}
\label{tab:schizophrenia_auc}
\begin{tabular}{llc}
\hline
Outcome & Method & AUC \\
\hline
Depression & DGNMM & 0.989 \\
           & GNMM-1 & 0.982 \\
           & GNMM-2 & 0.961 \\
           & ANN & 0.870 \\
           & GLMM Main Effects & 0.875 \\
           & GLMM + 3 Interactions & 1.000 \\
           & GLMM + 5 Interactions & 1.000 \\
           & MP & 0.886 \\
\hline
Anxiety    & DGNMM & 0.836 \\
           & GNMM-1 & 0.786 \\
           & GNMM-2 & 0.964 \\
           & ANN & 0.614 \\
           & GLMM Main Effects & 0.857 \\
           & GLMM + 3 Interactions & 0.786 \\
           & GLMM + 5 Interactions & 0.786 \\
           & MP & 0.743 \\
\hline
\end{tabular}
\end{table}

For the binary symptom classification analysis,
Table~\ref{tab:schizophrenia_auc} shows a more heterogeneous pattern.
For depression, DGNMM achieved an AUC of 0.989, exceeding GNMM-1
(0.982), GNMM-2 (0.961), ANN (0.870), the main-effects GLMM (0.875),
and MP (0.886). However, the GLMM specifications with three and five
interaction terms each achieved an AUC of 1.000. Given the small number
of subjects and the resulting number of holdout observations, these
near-perfect classification results should be interpreted cautiously.
For anxiety, DGNMM achieved an AUC of 0.836. GNMM-2 produced the
highest AUC of 0.964, followed by the main-effects GLMM at 0.857 and
DGNMM at 0.836. DGNMM nevertheless outperformed GNMM-1 (0.786), the
interaction-augmented GLMM specifications (0.786), MP (0.743), and ANN
(0.614). Thus, unlike the continuous-outcome results, DGNMM was not the best-performing method for either binary classification task.

Taken together, the continuous and binary analyses demonstrate that the relative performance of DGNMM differs across schizophrenia symptom prediction tasks. For continuous symptom severity, DGNMM achieved the lowest MSPE for both depression and anxiety, indicating improved prediction of individual symptom scores relative to the main-effects
GLMM and the shallower GNMM specifications. The improvement was
particularly pronounced for anxiety severity, for which DGNMM reduced
MSPE by approximately 48\% relative to the main-effects GLMM and 88\%
relative to GNMM-1.

The classification results provide a more heterogeneous picture. For
distinguishing high versus low depression symptoms, DGNMM achieved an
AUC of 0.989, although the interaction-augmented GLMM specifications
achieved AUC values of 1.000. For high versus low anxiety symptoms,
GNMM-2 achieved the highest AUC of 0.964, followed by the main-effects
GLMM at 0.857 and DGNMM at 0.836. Thus, DGNMM showed its clearest
predictive benefit when estimating continuous depression and anxiety
severity rather than when classifying symptoms into high- and
low-severity groups.

This distinction is consistent with the simulation results, where the
predictive benefit of increased nonlinear representation capacity was
more pronounced for continuous prediction than for binary
discrimination. In the present application, these findings suggest that
DGNMM may be particularly useful for predicting variation in continuous
schizophrenia symptom severity from longitudinal smartphone-derived
mobility information, while simpler or shallower alternatives may remain
competitive for binary symptom classification.

\subsection{Subject-Specific Heterogeneity Analysis}
\label{sec:heterogeneity}

We additionally examined the subject-specific component of DGNMM and
whether its magnitude was related to the predictive benefit of the model.
For each individual, subject-specific heterogeneity was summarized by the
magnitude of the estimated DGNMM subject-specific effect, $|\hat e_i|$.
Because the final observation from each subject was used as the
within-subject holdout observation, let $Y_i$ denote the observed response
for the held-out final observation of subject $i$, and let
$\widehat{Y}_{i,\mathrm{GLMM}}$ and
$\widehat{Y}_{i,\mathrm{DGNMM}}$ denote the corresponding predictions
obtained from the main-effects GLMM and DGNMM, respectively. The
subject-specific predictive gain of DGNMM relative to the main-effects
GLMM was defined as

\[
G_i =
\left(Y_i-\widehat{Y}_{i,\mathrm{GLMM}}\right)^2
-
\left(Y_i-\widehat{Y}_{i,\mathrm{DGNMM}}\right)^2,
\qquad i=1,\ldots,m.
\]

Here, $m$ denotes the number of subjects. Thus, $G_i>0$ indicates that
DGNMM produced a smaller squared prediction error than the GLMM for
subject $i$, $G_i<0$ indicates better prediction by the GLMM, and
$G_i=0$ indicates equal squared prediction errors for the two models.

\begin{figure}[ht]
\centering

\begin{minipage}{0.35\textwidth}
\centering
\includegraphics[width=\textwidth]
{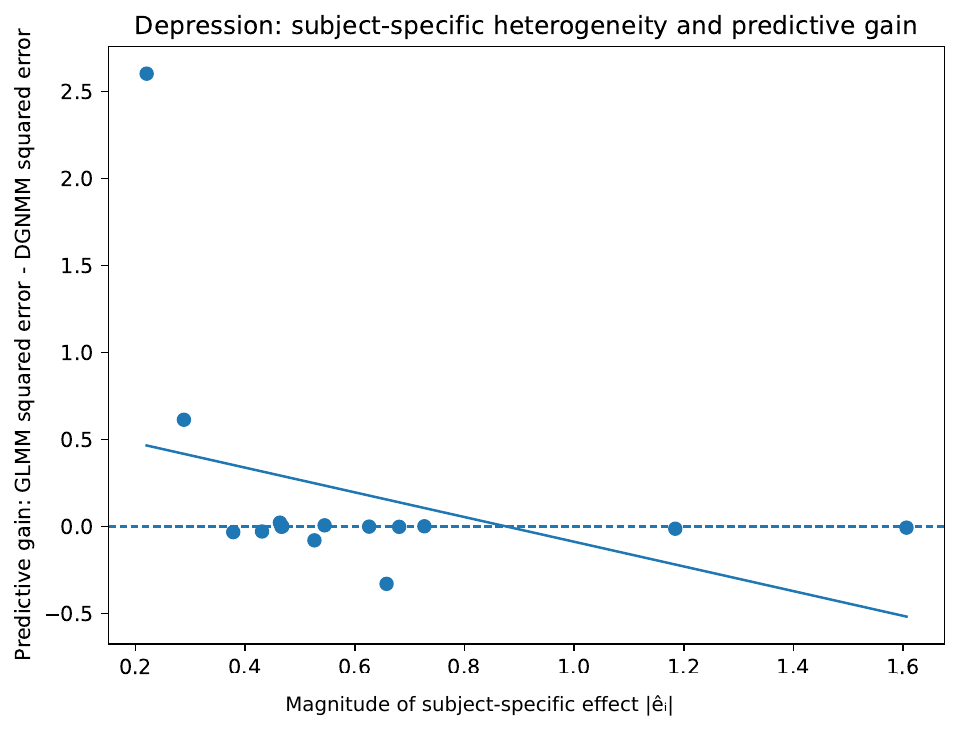}

\textbf{(a) Depression}
\end{minipage}
\hfill
\begin{minipage}{0.35\textwidth}
\centering
\includegraphics[width=\textwidth]
{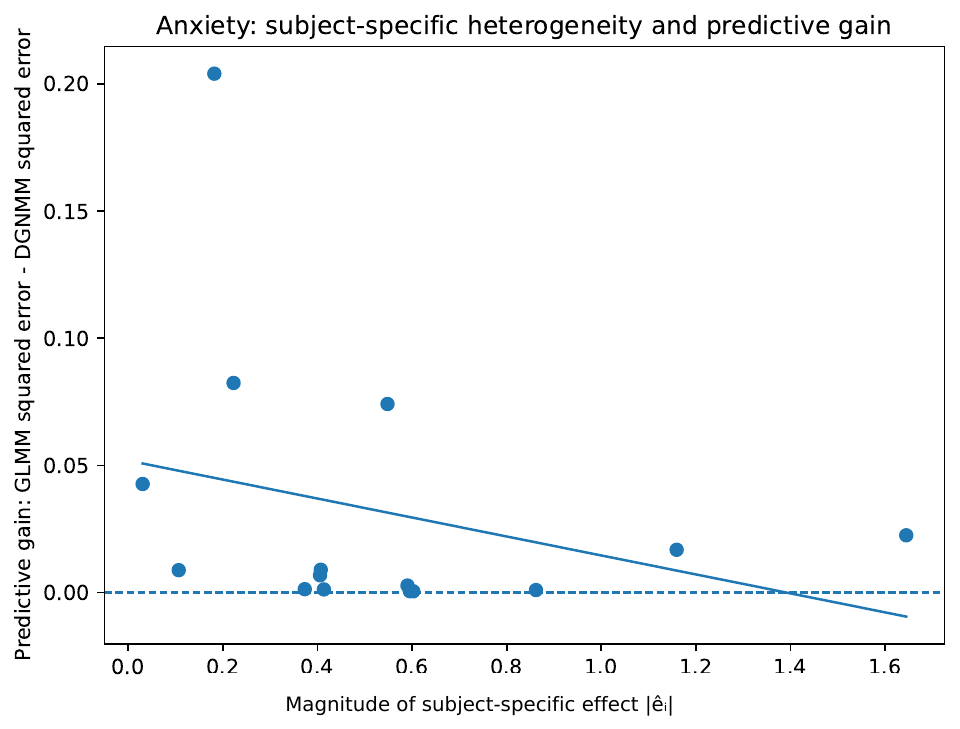}

\textbf{(b) Anxiety}
\end{minipage}

\caption{
Relationship between the magnitude of the estimated DGNMM
subject-specific effect $|\hat e_i|$ and subject-specific predictive
gain relative to the main-effects GLMM for (a) depression and
(b) anxiety. Positive predictive gain indicates a smaller squared
prediction error for DGNMM, whereas negative gain indicates a smaller
squared prediction error for GLMM. The Spearman correlations between
$|\hat e_i|$ and predictive gain were $\rho=-0.346$ ($p=0.206$) for
depression and $\rho=-0.407$ ($p=0.132$) for anxiety. Neither
association was statistically significant.
}
\label{fig:heterogeneity_gain}
\end{figure}
For depression, the subject-level predictive gain was positive for
6 of the 15 subjects, indicating that DGNMM produced a smaller squared
prediction error than the main-effects GLMM for these individuals.
Figure~\ref{fig:heterogeneity_gain}(a) shows substantial heterogeneity
in the magnitude of these gains. Most subjects had gains close to zero,
and both positive and negative values were observed. However, one subject
with a relatively small estimated subject-specific effect
($|\hat e_i|\approx0.2$) exhibited a predictive gain exceeding 2.5,
and another subject with a similarly small-to-moderate effect magnitude
showed a gain of approximately 0.6. In contrast, subjects with some of
the largest estimated subject-specific effects had predictive gains
close to zero. Thus, the improvement in depression prediction was not
uniform across individuals and appears to be driven in part by relatively
large improvements for a small subset of subjects.

For anxiety, the subject-level pattern was more consistent. DGNMM
produced a positive predictive gain for all 15 subjects in the
exploratory analysis. As shown in Figure~\ref{fig:heterogeneity_gain}(b),
most gains were relatively small and concentrated near zero, but several
subjects showed more substantial improvements. The largest predictive
gain was approximately 0.20 at $|\hat e_i|\approx0.2$, with additional
gains of approximately 0.08 and 0.07 observed at relatively small to
moderate effect magnitudes. Subjects with larger estimated
subject-specific effects generally showed only modest predictive gains.
Thus, compared with depression, the improvement for anxiety was more
consistent in direction across subjects but less dominated by a single
large gain.

We further examined whether the magnitude of the estimated
subject-specific effect was systematically associated with predictive
gain $G_i$. Because the application contained only 15 subjects, this
association was assessed using Spearman's rank correlation and is
interpreted as exploratory. For depression, the correlation between
$|\hat e_i|$ and $G_i$ was $\rho=-0.346$ ($p=0.206$), while for
anxiety it was $\rho=-0.407$ ($p=0.132$). Although both correlations
were negative, neither was statistically significant. The direction of
the estimated correlations is consistent with the visual patterns in
Figure~\ref{fig:heterogeneity_gain}: the largest predictive gains occur
among subjects with relatively small estimated subject-specific effects,
whereas larger values of $|\hat e_i|$ are not accompanied by
systematically larger gains. Therefore, the results do not support the
hypothesis that greater estimated subject-specific heterogeneity, as
measured by $|\hat e_i|$, necessarily corresponds to a greater
predictive benefit from DGNMM.

These findings provide additional context for the overall prediction
results. DGNMM achieved lower average prediction error than the
main-effects GLMM for both depression and anxiety, but the distribution
of that improvement across subjects differed between the two outcomes.
For depression, predictive gains were heterogeneous and concentrated
among a subset of individuals, whereas for anxiety the direction of the
gain was more consistent across subjects. Importantly, the magnitude of
the estimated subject-specific effect alone did not identify the subjects
with the greatest predictive improvement. Given the small sample of
15 subjects and the absence of statistically significant correlations,
these subject-level findings should be regarded as descriptive and
exploratory rather than evidence of a general association between
subject-specific heterogeneity and predictive gain.

\section{Conclusion}\label{sec6}

We proposed a Deep Generalized Neural Network Mixed Model (DGNMM) for
longitudinal prediction that combines increased nonlinear representation
capacity with regularized subject-specific effects. The framework is
estimated through prediction-oriented stochastic optimization, providing
a direct approach for jointly learning the population-level response
surface and persistent subject-specific heterogeneity.

The results demonstrate that increased nonlinear representation capacity
is not uniformly beneficial. In the simulations, GLMM performed best
when the population-level response was linear, whereas DGNMM provided
substantial predictive gains for continuous outcomes as moderate and
strong nonlinear structure was introduced. The advantage was less
pronounced for binary outcomes, where GLMM and shallower GNMM
alternatives remained competitive across many settings. The schizophrenia
digital-phenotyping application showed a similar task-dependent pattern:
DGNMM achieved the lowest MSPE for continuous depression and anxiety
prediction, while alternative models remained competitive or preferable
for binary symptom classification. Exploratory subject-level analyses
further showed that the magnitude of the estimated subject-specific
effect was not significantly associated with predictive gain.

Overall, DGNMM provides a flexible prediction framework for longitudinal
problems in which nonlinear population-level structure and
subject-specific heterogeneity may coexist, while the results emphasize
that additional model complexity should be guided by the prediction
setting rather than assumed to be universally advantageous. There is an opportunity for future work to extend the framework to
richer subject-specific structures such as random slopes and time-varying
effects, and to other longitudinal outcomes including counts, time-to-event
responses, and multivariate data. Additionally, developing procedures for
architecture and regularization selection and evaluating the framework in
larger longitudinal studies are important directions for future research.
\newpage

\bmsubsection*{Financial Disclosure}

None reported.

\bmsubsection*{Conflicts of Interest}

The authors declare no conflict of interest.

\bmsubsection*{Supporting Information}

The data supporting the findings of this study are available at
\url{https://onlinelibrary.wiley.com/doi/10.1111/biom.13615}. The code used in this study is available from the corresponding author upon reasonable request.

\end{document}